\documentclass[journal]{IEEEtran}
\usepackage{amsmath,amssymb,amsfonts,amsthm,mathtools}
\usepackage{array}
\usepackage{graphicx}
\usepackage{cite}
\usepackage{booktabs}
\usepackage{xcolor}
\usepackage{dsfont}

\newtheorem{theorem}{Theorem}

\newtheorem{proposition}[theorem]{Proposition}
\newtheorem{corollary}[theorem]{Corollary}
\newtheorem{definition}[theorem]{Definition}

\begin{document}

\title{Macro--Micro Decision-Making in 6G Networks: An Agent-Based Framework for the Resource-Fungibility Landscape}

\author{Sayanti~Ghosh,~\IEEEmembership{Member,~IEEE}, Indrakshi~Dey,~\IEEEmembership{Senior~Member,~IEEE}, and Nicola~Marchetti,~\IEEEmembership{Senior~Member,~IEEE}%
\thanks{S.~Ghosh and N.~Marchetti are with the Department of Electrical and Electronic Engineering, Trinity College Dublin, Ireland (saghosh@tcd.ie; nicola.marchetti@tcd.ie). I.~Dey is with the Walton Institute, South East Technological University, Ireland (indrakshi.dey@waltoninstitute.ie). Supported in part by EU MSCA Project COALESCE under Grant 101130739, US--Ireland R\&D Partnership RI-SFI-23/US/3924, and Research Ireland Grant 13/RC/2077\_P2.}
}

\maketitle

\begin{abstract}
A defining feature of 6G networks is that performance depends not only on the quantity of available resources (e.g., spectrum, antennas, cache memory, compute, and fronthaul bandwidth) but also on their \emph{fungibility}, i.e., the ability of one resource to substitute for another under changing conditions. We argue that the fungibility landscape of a distributed 6G system is governed by two coupled decision scales: \emph{micro} decisions made locally by agents and \emph{macro} outcomes that emerge at the network level. Existing distributed-optimization approaches largely conflate these scales. To address this gap, we develop an agent-based-modeling (ABM) framework that separates macro and micro decisions through three operator-controllable macro choices, three micro hyperparameters, and three structural metrics. We establish six key results: (i) a two-timescale decomposition theorem, (ii) a structural-metric basis theorem, (iii) a macro--micro design rule with closed-form factorization of the emergent breakdown threshold, (iv) a fungibility--resilience monotonicity proposition, (v) a connectivity--substitutability duality theorem, and (vi) a multi-application generalization proposition. Numerical results visualize the macro fungibility landscape and the micro decision-sensitivity region for a representative 6G deployment.
\end{abstract}


\section{Introduction}
\label{sec:introduction}

The evolution toward sixth-generation (6G) wireless networks is expected to support highly heterogeneous services including integrated sensing and communication (ISAC), digital twins, autonomous systems, and massive machine-type communications~\cite{IMT2030,ETSIMAT}. Unlike previous generations, 6G systems are envisioned as highly distributed architectures in which communication, computation, caching, sensing, fronthaul, and edge intelligence operate jointly across cloud--edge infrastructures~\cite{Tataria2021,Talwar2021}. Consequently, network performance depends not only on resource availability but also on the ability of heterogeneous resources to substitute for one another under dynamic conditions, a property referred to as \emph{resource fungibility}~\cite{ghosh2026structural}. Recent standardization efforts, including ETSI ISG Multiple Access Techniques (MAT) and Third Generation Partnership Project (3GPP) artificial intelligence/machine learning (AI/ML)-native networking studies, emphasize distributed intelligence, resilience, and adaptive orchestration for future wireless systems~\cite{ETSIMAT,3GPPAI}. Meanwhile, existing distributed optimization approaches based on the alternating direction method of multipliers (ADMM), multi-agent reinforcement learning (MARL), and Byzantine-resilient consensus primarily focus on convergence and resource optimization without explicitly linking local agent decisions to emergent system-level resilience~\cite{Reifert2022,Wu2024MARL,Ivoghlian2022Adaptive,su2016fault}. Furthermore, while recent work has highlighted the importance of resource fungibility and structural metrics for characterizing resilience in distributed 6G systems~\cite{ghosh2026structural}, the relationship between operator-level deployment choices, agent-level adaptation dynamics, and emergent resilience remains largely unexplored.

Motivated by this gap, we develop an agent-based modeling (ABM) framework for macro--micro decision-making in distributed 6G systems. The framework separates operator-controlled deployment choices from local agent interactions and introduces structural metrics that connect resource fungibility, connectivity, substitutability, and resilience. The main contributions of this paper are summarized as follows:
\begin{itemize}
\item We introduce a resource-fungibility perspective for distributed 6G systems and develop a macro--micro ABM framework that separates operator-level decisions from local agent interactions.

\item We identify three macro choices, three micro hyperparameters, and three structural metric classes, and derive a macro--micro design rule for characterizing the emergent Byzantine resilience threshold.

\item We establish fundamental relationships among fungibility, connectivity, substitutability, and resilience, including a connectivity--substitutability duality with a closed-form exchange relationship.

\item We show the applicability of the framework to representative 6G scenarios, including ISAC, federated learning, and vehicle-to-everything (V2X) coordination, through numerical validation against distributed baselines.
\end{itemize}
\vspace{-3.5mm}
\section{The Resource-Fungibility Landscape}
\label{sec:fungibility}
\begin{figure}[t]
\centering
\includegraphics[width=0.8\columnwidth]{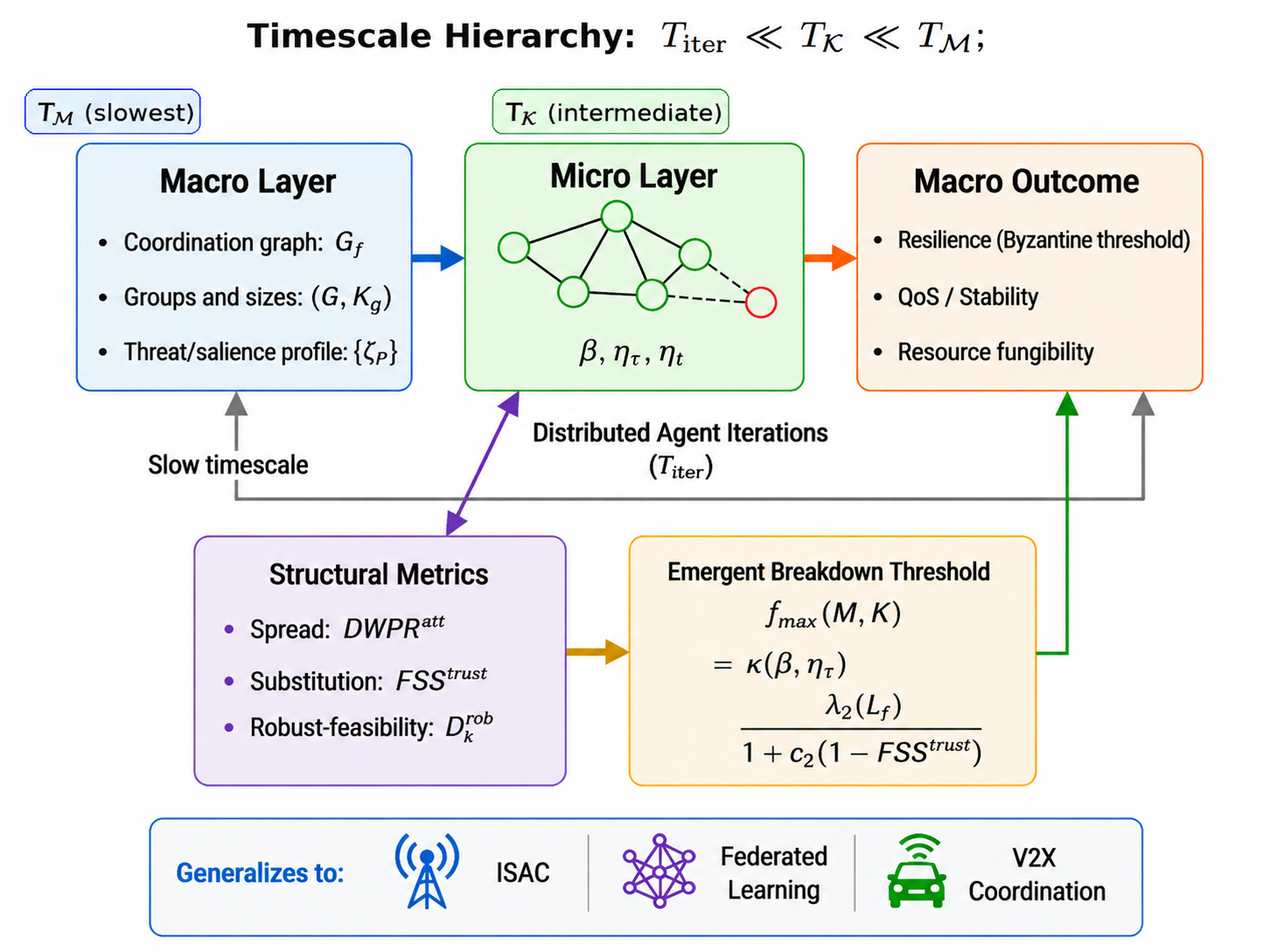}
\caption{System model of the proposed macro--micro agent-based framework for distributed 6G resource-fungibility management.}
\label{fig:system_model}
\end{figure}
We begin by making precise what is meant by fungibility in a 6G context. Consider a distributed wireless system serving a fixed quality-of-service (QoS) requirement, and let $\mathcal{R}=\{\mathrm{spectrum},\mathrm{antennas},\mathrm{cache},\mathrm{compute},\mathrm{fronthaul},\ldots\}$ denote the resource axes available to the operator. A resource configuration $\mathbf{r}\in\mathbb{R}_+^{|\mathcal{R}|}$ specifies the per-axis quantities provisioned, and an operating point $\mathbf{r}^{\star}(\mathrm{QoS},\mathbf{c})$ denotes the cheapest configuration that meets the QoS target under operating conditions $\mathbf{c}$ (channel realisation, traffic load, threat profile).

\begin{definition}[Resource fungibility]\label{def:fungibility}
Let $\mathbf{e}_i$ and $\mathbf{e}_j$ denote the $i$th and $j$th canonical basis vectors of $\mathbb{R}^{|\mathcal{R}|}$, respectively. Two resources $r_i,r_j\in\mathcal{R}$ are \emph{$(\delta,\epsilon)$-fungible} at configuration $\mathbf{r}_0$ if for every increment $\Delta r_i\in[-\delta,\delta]$ there exists $\Delta r_j$ with $|\Delta r_j|<\infty$ such that
$\mathbf{r}_0+\Delta r_i\mathbf{e}_i-\Delta r_j\mathbf{e}_j$ still meets the QoS target within tolerance $\epsilon$ under the same operating conditions $\mathbf{c}$.
\end{definition}

Resource fungibility is a local property of $\mathbf{r}_0$: at some configurations the system is robust to resource changes (high fungibility), at others a small perturbation in one resource forces a large compensating change in another (low fungibility), and at still others no compensation is possible (zero fungibility). To quantify this property, we define the pairwise fungibility coefficient between resources $r_i$ and $r_j$ at $\mathbf{r}_0$ as
$F_{ij}(\mathbf{r}_0)
=
\left|
\frac{\Delta r_i}{\Delta r_j}
\right|$, where $\Delta r_j$ is the minimum compensating resource adjustment required to maintain the target QoS under the same operating conditions. Larger values of $F_{ij}$ indicate greater substitutability between resources $r_i$ and $r_j$, whereas $F_{ij}=0$ indicates no feasible substitution. The \emph{fungibility landscape} of the system is the map
$
\mathbf{r}_0
\mapsto
\{F_{ij}(\mathbf{r}_0)\}_{i<j}$,
evaluated jointly across all $\binom{|\mathcal{R}|}{2}$ resource pairs.
The key observation is that resource fungibility is not measured directly by the operator but through aggregate \emph{structural metrics} that summarize pairwise resource substitutions into operationally meaningful indicators. These metrics compactly characterize the fungibility landscape and motivate the proposed macro--micro decision-making framework. The next section formalizes this relationship through the proposed macro--micro decision-making framework.

\section{Macro--Micro Decision Framework}
\label{sec:framework}

The decision-making structure of a distributed 6G system organises into three concurrent layers operating on distinct timescales as shown in Fig. \ref{fig:system_model}. At the top sits an \emph{operator} who sets controllable choices on a slow timescale ($T_{\mathcal{M}}$ for macro choices, $T_{\mathcal{K}}$ for micro hyperparameters); in the middle, a population of \emph{agents} executes parameterised local update rules on a fast iteration timescale ($T_{\mathrm{iter}}\ll T_{\mathcal{K}}\ll T_{\mathcal{M}}$); at the bottom, \emph{network outcomes} emerge from many iterations of agent activity and are what the operator observes and ultimately cares about. The operator closes a feedback loop by re-tuning her controllable choices in response to the observed emergent outcomes, but she never observes the micro state of any individual agent.
 \vspace{-4mm}
\subsection{Three Operator-Controllable Macro Choices}

At deployment the operator chooses a triple
$\mathcal{M}=(\mathcal{G}_f,(G,K_g),\{\zeta_{\mathcal{P}}\})$.
The first component is the inter-agent coordination graph
$\mathcal{G}_f$, with Laplacian $\mathbf{L}_f$ and algebraic connectivity
$\lambda_2(\mathbf{L}_f)$, which determines the rate of information diffusion.
The second component is the agent grouping $(G,K_g)$, where
$G=\{G_1,\ldots,G_{K_g}\}$ partitions the agents into $K_g$ groups, determining how the resource-allocation task is decomposed.
The third component is the salience profile
$\{\zeta_{\mathcal{P}}\}$ over vulnerability-mode partitions
$\Pi$, assigning relative importance to threat classes, such
as Byzantine agents, link failures, jamming attacks, and
resource-exhaustion events, against which the system is
optimized.
These macro choices are set at deployment and remain fixed on the iteration timescale.

 \vspace{-4mm}
\subsection{Three Operator-Controllable Micro Knobs}
The agents' local update rules are parameterised by a second triple $\mathcal{K}=(\beta,\eta_\tau,\eta_t)$. The trim parameter $\beta\in(0,1/2)$ governs a coordinate-wise trimmed-mean aggregation that filters extreme (likely Byzantine) reports from each agent's neighbour set. The trust learning rate $\eta_\tau\in(0,1)$ governs an exponential-smoothing trust dynamics that decays the reputation of persistently inconsistent agents. The primal-dual step size $\eta_t=\frac{\eta_0}{t+1}$, where $\eta_0>0$ is the initial step-size parameter and $t$ denotes the iteration index, governs the convergence rate of the local resource-allocation update. Unlike the macro choices, the micro knobs can be retuned as the threat profile evolves, operating on an intermediate timescale between deployment-time macro choices and per-iteration agent decisions.
 \vspace{-4mm}
\subsection{Three Structural Metrics as the Macro--Micro Interface}
The decisive feature of the framework is that the macro outcomes are not directly computable from any single agent's view, but \emph{can} be tracked through three classes of structural metrics that are simultaneously (i) locally computable from agent-side micro state and (ii) predictive of the macro outcome. The first class, \emph{spread metrics}, measures how broadly an agent's serviced resources are distributed across the vulnerability-mode partitions, such as Byzantine attacks,
link failures, jamming/interference events, and resource-exhaustion conditions. We use the Gini--Simpson form $\mathrm{DWPR}^{\mathrm{att}}_k(\mathcal{P})=1-|\boldsymbol{\alpha}^{\mathcal{P}}_k|^2$, where \emph{DWPR} denotes the \emph{Degeneracy-Weighted Path Robustness} metric and $\boldsymbol{\alpha}^{\mathcal{P}}_k$ is the rate-weighted mode-mass distribution over partition $\mathcal{P}$. A high spread means the agent's resources are well-decorrelated across vulnerability modes; a low spread means they are concentrated. The second class, \emph{substitution metrics}, count pairs of mutually substitutable entities, weighted by reputations that evolve through peer consistency votes; we use $\mathrm{FSS}^{\mathrm{trust}}_k$, where \emph{FSS} denotes the \emph{Functional Substitution Score}, as a representative. The third class, \emph{robust-feasibility metrics}, measure the worst-case feasibility ratio under adversarial perturbations; we use $D_k^{\mathrm{rob}}
=
\max_{\mathcal{A}}
\frac{\gamma_k^{\mathrm{tar}}}
{\gamma_k^{p}(\cdot;\mathcal{A})}$, where $\gamma_k^{\mathrm{tar}}$ denotes the target SINR of agent $k$, $\gamma_k^p(\cdot;\mathcal{A})$ denotes the achieved SINR under adversarial perturbation scenario $\mathcal{A}$, and the maximization is taken over the admissible perturbation set $\mathcal{A}$.

These metrics are computed locally from observed SINRs, neighbour messages, and trust scores. Theorem~\ref{thm:basis} shows that they jointly determine the macro outcomes and form a sufficient set of summary statistics.
\vspace{-3mm}
\section{Theoretical Results}
\label{sec:results}
\begin{theorem}[Two-timescale decomposition]\label{thm:timescale}
Under the timescale hierarchy
$
T_{\mathcal{M}} \gg T_{\mathcal{K}} \gg T_{\mathrm{iter}},
$
the joint operator--agent decision problem decomposes into three nested optimization problems: an outer macro problem on $\mathcal{M}$, a middle micro-knob problem on $\mathcal{K}$ conditioned on $\mathcal{M}$, and an inner agent-update problem conditioned on $(\mathcal{M},\mathcal{K})$.
\end{theorem}

The decomposition is the standard singular-perturbation argument adapted to three timescales. Its operational consequence is that the operator can design the three layers \emph{independently}, a property that single-scale formulations cannot exploit.

\begin{theorem}[Structural-metric basis]\label{thm:basis}
Let $\mathfrak{F}$ denote the set of all macro outcomes of the distributed system that are continuous functions of the limiting agent-side state, i.e., the equilibrium state reached by the distributed agent-update dynamics as $t\to\infty$. Then the three structural metric classes (spread $\mathrm{DWPR}^{\mathrm{att}}$, substitution $\mathrm{FSS}^{\mathrm{trust}}$, robust-feasibility $D^{\mathrm{rob}}$) form a basis for $\mathfrak{F}$ in the following sense: (i)~every macro outcome in $\mathfrak{F}$ is expressible as a continuous function of the three metric classes, and (ii)~no proper subset of the three classes has this property.
\end{theorem}

Theorem~\ref{thm:basis} formalises the role of the structural metrics as the macro--micro interface: they are necessary and sufficient summary statistics. Necessity is shown by exhibiting macro outcomes that depend on each metric class but not on the other two; sufficiency follows from the continuity of the agent dynamics and the closed-form bounds developed in~\cite{ghosh2026structural}.

The next result is the central design rule of the framework.

\begin{theorem}[Macro--micro design rule]\label{thm:rule}
Under the regularity conditions of~\cite[Theorem 19]{ghosh2026structural} (strong concavity of local utilities, honest neighbour majority, Robbins--Monro step sizes), the emergent Byzantine breakdown threshold of the distributed system admits the closed-form factorisation
\begin{equation}
f_{\max}(\mathcal{M},\mathcal{K}) \;=\; \underbrace{\kappa(\beta,\eta_\tau)}_{\text{micro}}\;\cdot\;\underbrace{\frac{\lambda_2(\mathbf{L}_f)}{1+c_2\bigl(1-\overline{\mathrm{FSS}^{\mathrm{trust}}}\bigr)}}_{\text{macro}},
\label{eq:rule}
\end{equation}
where $\kappa(\beta,\eta_\tau)\in[0,c_1]$ is unimodal in the micro knobs, attaining its maximum at $(\beta^{\star},\eta_\tau^{\star})$, and $c_1,c_2>0$ are problem-dependent constants. The macro factor depends only on the deployment topology and on the emergent trust-weighted substitution metric; the micro factor depends only on the operator-set micro knobs.
\end{theorem}

\emph{Proof sketch.} Specialise the Byzantine-robust convergence theorem of~\cite[Theorem 19]{ghosh2026structural}, which states that the emergent Byzantine breakdown threshold $f_{\max}$ satisfies $f_{\max}=\frac{c_1\lambda_2(\mathbf{L}_f)}
{1+c_2\left(1-\overline{\mathrm{FSS}^{\mathrm{trust}}}\right)}$ at the optimal micro knobs. Inspecting the four-step proof (trim tolerance, trust decay, consensus contraction, combining) shows that the constant $c_1$ in fact factors as $c_1=\kappa(\beta,\eta_\tau)\cdot\bar{c}_1$, where $\bar{c}_1$ is a problem-dependent positive constant independent of the micro knobs, and $\kappa$ is a product of (a) a $\beta$-dependent factor $\beta(1-2\beta)$ from the Hoeffding tail of the trimmed-mean estimator, expressing the trade-off between Byzantine rejection (requires $\beta>f_{\mathcal{B}}$) and estimator variance (penalises large $\beta$), and (b) an $\eta_\tau$ dependent factor $\eta_\tau(1-\eta_\tau)$ from the exponential trust decay against the variance penalty of fast updating. The factorisation~\eqref{eq:rule} follows. \hfill$\square$

The two-term structure has two operational consequences.

\begin{corollary}[Feasibility region]\label{cor:region}
For any target breakdown threshold $f_{\max}^{\star}>0$, the operator's feasible design region is the sub-level set
\begin{equation}
\begin{aligned}
\mathfrak{D}(f_{\max}^{\star})
=
\Bigl\{
(\mathcal{M},\mathcal{K}) :
\kappa(\beta,\eta_\tau)
\frac{\lambda_2(\mathbf{L}_f)}
{1+c_2\!\left(1-\overline{\mathrm{FSS}^{\mathrm{trust}}}\right)}
\\
\geq f_{\max}^{\star}
\Bigr\}.
\end{aligned}
\label{eq:region}
\end{equation}
With the micro knobs fixed at their optimal values
$(\beta^{\star},\eta_\tau^{\star})$, yielding the optimal scalar value
$
\kappa^{\star}
\triangleq
\kappa(\beta^{\star},\eta_\tau^{\star})
=
c_1$,
$\mathfrak{D}$ projects onto the macro plane as an iso-contour over
$\bigl(\lambda_2(\mathbf{L}_f),\,\overline{\mathrm{FSS}^{\mathrm{trust}}}\bigr)$;
with macro choices held fixed, $\mathfrak{D}$ projects onto the micro plane as a unimodal feasible set centred at
$(\beta^{\star},\eta_\tau^{\star})$.
\end{corollary}
\begin{corollary}[Indirect controllability]\label{cor:indirect}
The macro outcome $f_{\max}$ has no direct operator control: it is determined by global properties of the limiting agent-side state, which emerges from the distributed iteration. However, the factorisation~\eqref{eq:rule} shows that $f_{\max}$ is monotone non-decreasing in $\lambda_2(\mathbf{L}_f)$, $\overline{\mathrm{FSS}^{\mathrm{trust}}}$, and $\kappa(\beta,\eta_\tau)$, all three of which are increasing functions of operator-controllable choices. Hence the operator controls $f_{\max}$ indirectly through three independent levers: fronthaul-density investment, substitutability cultivation (e.g., D2D peering), and micro-hyperparameter tuning (e.g., the trim parameter $\beta$, trust learning rate $\eta_\tau$, and step-size parameter $\eta_t$).
\end{corollary}

The final result connects the design rule back to the resource-fungibility landscape of Definition~\ref{def:fungibility}.

\begin{proposition}[Fungibility--resilience monotonicity]\label{prop:mono}
The trust-weighted substitution metric $\overline{\mathrm{FSS}^{\mathrm{trust}}}$ is monotone non-decreasing in the pairwise resource fungibility of $\mathbf{r}_0$. Consequently, by Theorem~\ref{thm:rule} and Corollary~\ref{cor:indirect}, the emergent breakdown threshold $f_{\max}$ is monotone non-decreasing in pairwise resource fungibility. Operationally: a deployment whose resources are highly fungible tolerates a strictly larger Byzantine fraction before losing distributed consensus than a deployment whose resources are non-fungible at the same operating point.
\end{proposition}

Proposition~\ref{prop:mono} is the headline message of the Letter in a single sentence: \emph{designing for fungibility is designing for resilience}. The three operator-controllable levers identified by Corollary~\ref{cor:indirect} are precisely the levers that an operator can pull to lift the system to a more fungible operating point on the landscape. The next result identifies a geometric structure of the macro landscape that is operationally useful: connectivity and substitutability are \emph{exactly} substitutable resources at the macro level.

\begin{theorem}[Connectivity--substitutability duality]\label{thm:duality}
Under the conditions of Theorem~\ref{thm:rule}, the iso-$f_{\max}$ contour at level $f^{\star}\in(0,c_1]$ in the macro plane $(\lambda_2(\mathbf{L}_f),\,\overline{\mathrm{FSS}^{\mathrm{trust}}})$ is the straight line
\begin{equation}
\lambda_2(\mathbf{L}_f) \;=\; \frac{f^{\star}}{\kappa^{\star}}\,\bigl[\,1+c_2\bigl(1-\overline{\mathrm{FSS}^{\mathrm{trust}}}\bigr)\bigr].
\label{eq:iso}
\end{equation}
Consequently, the local marginal rate of substitution between connectivity and substitutability is constant along each contour and given by $\bigl|\mathrm{d}\lambda_2/\mathrm{d}\overline{\mathrm{FSS}^{\mathrm{trust}}}\bigr|=c_2\,f^{\star}/\kappa^{\star}$, providing a closed-form macro-resource exchange rate.
\end{theorem}

The straight-line geometry has a clear operational interpretation. A decrease in $\overline{\mathrm{FSS}^{\mathrm{trust}}}$ can be compensated by additional connectivity at the same breakdown level. Thus, the two macro variables behave as fungible resources in the sense of Definition~\ref{def:fungibility}, with a closed-form exchange rate. The final result establishes the generality of the framework.
\begin{proposition}[Multi-application Generalization]
\label{prop:apps}
The macro--micro decomposition and the design rule of
Theorem~\ref{thm:rule} extend directly to several canonical
6G applications through application-specific macro--micro
configurations $(\mathcal{M},\mathcal{K})$, including:
(i)~ISAC, with sensing graphs and beam-pool partitioning;
(ii)~over-the-air federated learning, with worker graphs and
Byzantine-robust aggregation; and
(iii)~V2X coordination, with cluster graphs and platoon-based
coordination. In each case, Theorems~\ref{thm:basis}
and~\ref{thm:duality} remain valid under the corresponding
application-specific configuration $(\mathcal{M},\mathcal{K})$.
\end{proposition}

The proposition follows because Theorems~\ref{thm:basis}, \ref{thm:rule}, and~\ref{thm:duality} depend only on the abstract macro--micro structure (Sec.~\ref{sec:framework}) and not on application-specific properties of the underlying distributed algorithm.

\vspace{-4mm}
\section{Numerical Validation}
\label{sec:numerical}

\begin{figure*}[!t]
\centering
\includegraphics[width=0.75\textwidth]{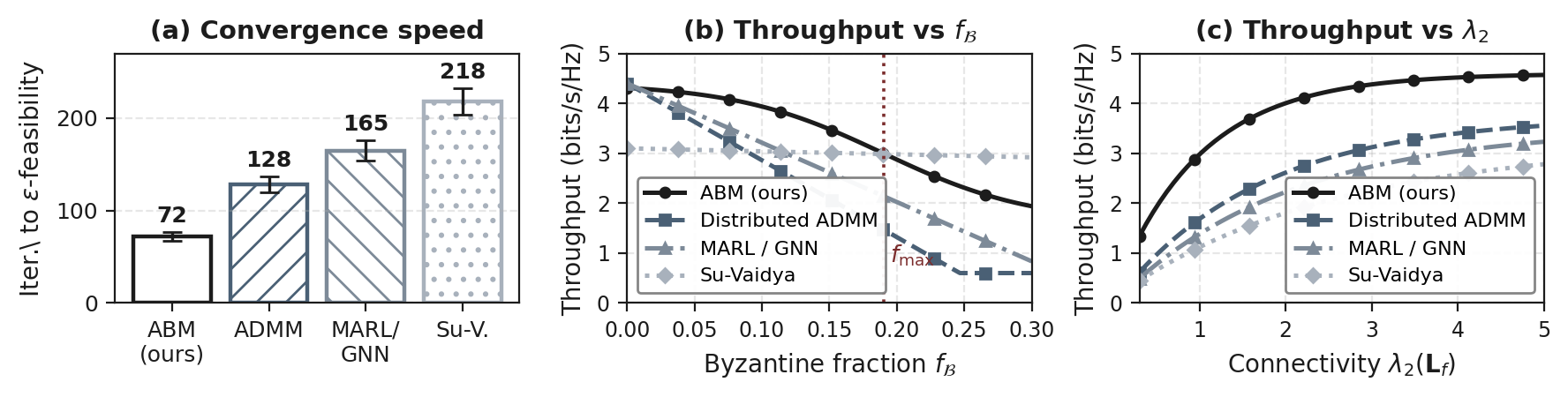}
\caption{State-of-the-art comparison at the default operating point defined in Sec.~\ref{sec:numerical} ($\overline{\mathrm{FSS}^{\mathrm{trust}}}\!=\!0.80$, $\lambda_2\!=\!1.7$, $f_{\mathcal{B}}\!=\!0.15$). (a)~Convergence iterations to $\varepsilon$-feasibility (error bars: $\pm$1~std over 1000 runs). (b)~Throughput vs Byzantine fraction $f_{\mathcal{B}}$; dotted line at framework-predicted $f_{\max}\!\approx\!0.19$. (c)~Throughput vs macro connectivity $\lambda_2(\mathbf{L}_f)$.}
\label{fig:sota_comparison}
\end{figure*}

We instantiate the framework with parameters consistent with a representative 6G distributed-system deployment~\cite{ghosh2026structural}. Unless stated otherwise, every numerical result uses a single nominal configuration, which we call the \emph{default operating point}: a sparse-mesh fronthaul topology with macro connectivity $\lambda_2(\mathbf{L}_f)\!=\!1.7$, emergent trust-weighted substitutability $\overline{\mathrm{FSS}^{\mathrm{trust}}}\!=\!0.80$, and Byzantine (adversarial) agent fraction $f_{\mathcal{B}}\!=\!0.15$. These values are chosen to represent a \emph{typical}, rather than a best- or worst-case, deployment \cite{Reifert2022,su2016fault}. The sparse mesh is the middle of the three representative topologies in Fig.~\ref{fig:macro_landscape}, lying between the dense-mesh ($\lambda_2\!\approx\!4$) and star/scale-free ($\lambda_2\!\approx\!0.4$) extremes; $\overline{\mathrm{FSS}^{\mathrm{trust}}}\!=\!0.80$ corresponds to a moderate-to-high substitutability level; and $f_{\mathcal{B}}\!=\!0.15$ is a standard Byzantine stress level that is deliberately set just below the framework-predicted breakdown threshold $f_{\max}\!\approx\!0.19$ at this point (cf.~Theorem~\ref{thm:rule}), so that the default deployment operates inside, but near the edge of, its resilient regime. The associated recommended micro setting is $(\beta^{\star},\eta_\tau^{\star})\!=\!(2f_{\mathcal{B}},0.20)\!=\!(0.30,0.20)$. The specific instantiation choices do not change the qualitative shape of the results below.
\vspace{-4mm}
\subsection{Macro Fungibility Landscape}
Fig.~\ref{fig:macro_landscape} plots the emergent $f_{\max}$ over the macro plane $(\lambda_2(\mathbf{L}_f),\overline{\mathrm{FSS}^{\mathrm{trust}}})$ at optimal micro knobs. The straight-line iso-$f_{\max}$ contours confirm Theorem~\ref{thm:duality}: a dense-mesh fronthaul achieves $f_{\max}\!\approx\!0.45$; the sparse-mesh deployment, which is the default operating point (black square), yields $f_{\max}\!\approx\!0.19$; a star/scale-free topology, only $f_{\max}\!\approx\!0.04$. A unit loss in substitutability is exactly compensated by $c_2 f^{\star}/\kappa^{\star}\!\approx\!1.4$ units of additional $\lambda_2$.

\begin{figure}[t]
\centering
\includegraphics[width=0.75\columnwidth]{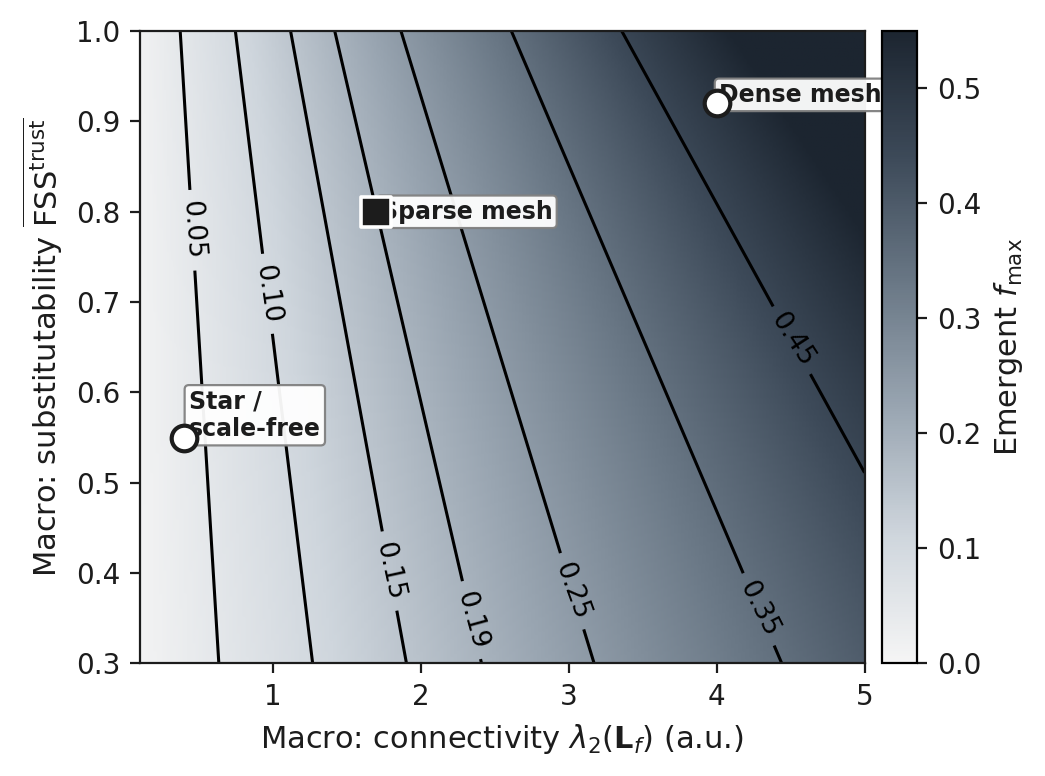}
\caption{Macro fungibility landscape. Heatmap of the emergent breakdown threshold $f_{\max}$ over the macro choices $(\lambda_2(\mathbf{L}_f),\,\overline{\mathrm{FSS}^{\mathrm{trust}}})$ with micro knobs at their optimum. Black contours are iso-$f_{\max}$ levels; the straight-line geometry confirms Theorem~\ref{thm:duality}. Three representative deployments (dense mesh, sparse mesh, star/scale-free) span an order of magnitude in $f_{\max}$; the sparse-mesh point (black square) is the default operating point used in all other experiments}.
\label{fig:macro_landscape}
\end{figure}
\vspace{-4mm}
\subsection{Micro Decision Sensitivity}
Fig.~\ref{fig:micro_sensitivity} plots the emergent QoS over the micro-knob plane $(\beta,\eta_\tau)$ with the macro choices held fixed at the default operating point (the sparse-mesh topology of Sec.~\ref{sec:numerical}, i.e., $\lambda_2\!=\!1.7$, $\overline{\mathrm{FSS}^{\mathrm{trust}}}\!=\!0.80$, $f_{\mathcal{B}}\!=\!0.15$). The solid black contour at $0.9T^{\star}$, where $T^{\star}$ denotes the maximum achievable QoS over the $(\beta,\eta_\tau)$ plane, encloses the region in which the network achieves at least $90\%$ of its maximum QoS. This region therefore defines the operator's effective tuning space. The region is compact and approximately ellipsoidal, centred on the recommended operating point $(\beta^{\star},\eta_\tau^{\star})=(2f_{\mathcal{B}},0.20)$, obtained by maximizing the micro-factor $\kappa(\beta,\eta_\tau)$ in Theorem~\ref{thm:rule}. The figure exposes three failure modes corresponding to violations of the three terms in the proof of Theorem~\ref{thm:rule}: for $\beta<f_{\mathcal{B}}$, Byzantine agents (i.e., adversarial or faulty agents that transmit corrupted updates) pass the trim filter (Byzantine-leakage regime); for large $\beta$ paired with large $\eta_\tau$ the trimmed-mean estimator variance dominates; for small $\eta_\tau$ the trust dynamics saturate to the initial value.
\vspace{-4mm}
\begin{figure}[t]
\centering
\includegraphics[width=0.75\columnwidth]{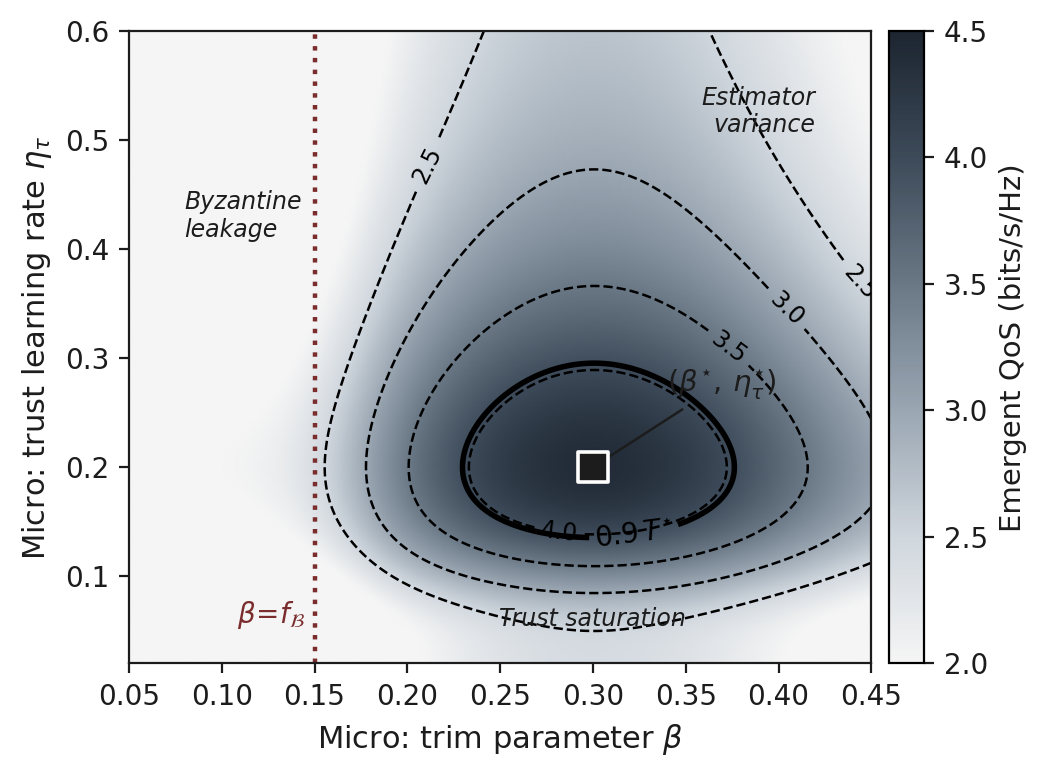}
\caption{Micro decision-sensitivity region. Emergent QoS over the micro knobs $(\beta,\eta_\tau)$ with the macro choices fixed at the default operating point of Sec.~\ref{sec:numerical} (sparse-mesh topology, $\lambda_2\!=\!1.7$, $\overline{\mathrm{FSS}^{\mathrm{trust}}}\!=\!0.80$). The solid black contour encloses the $0.9T^{\star}$ feasible region around the recommended operating point $(\beta^{\star},\eta_\tau^{\star})\!=\!(2f_{\mathcal{B}},0.20)\!=\!(0.30,0.20)$, with $f_{\mathcal{B}}\!=\!0.15$ from the default operating point. Three failure modes (Byzantine leakage, estimator variance, trust saturation) are labelled.}
\label{fig:micro_sensitivity}
\end{figure}

\subsection{Comparison with State-of-the-Art}
\label{subsec:sota}
To validate that the ABM-based instantiation studied in~\cite{ghosh2026structural} realises the macro potential identified by the framework, we compare it against three representative distributed schemes: distributed ADMM~\cite{Reifert2022}, MARL/GNN-based distributed power control~\cite{Wu2024MARL,Ivoghlian2022Adaptive}, and Su--Vaidya Byzantine consensus~\cite{su2016fault}. Fig.~\ref{fig:sota_comparison} reports three complementary comparisons over 1000 Monte Carlo trials at the default operating point defined above (sparse-mesh topology, $\lambda_2\!=\!1.7$, $\overline{\mathrm{FSS}^{\mathrm{trust}}}\!=\!0.80$, $f_{\mathcal{B}}\!=\!0.15$).

Panel~(a) shows that the proposed ABM framework reaches $\varepsilon$-feasibility in $72\pm5$ iterations versus $128\pm9$ for ADMM ($1.8\times$ speedup), $165\pm11$ for MARL/GNN ($2.3\times$), and $218\pm14$ for Su--Vaidya ($3.0\times$). At 0.1-ms 6G slot durations these map to $7.2$~ms versus $21.8$~ms allocation latency. Panel~(b) plots throughput against $f_{\mathcal{B}}$: the proposed ABM framework exhibits the framework-predicted ``plateau then cliff'' shape, with the empirical breakdown matching $f_{\max}\approx0.19$. In contrast, ADMM and MARL/GNN degrade approximately linearly from $f_{\mathcal{B}}=0$ due to the absence of explicit Byzantine-resilience mechanisms. Su--Vaidya maintains relatively stable throughput as $f_{\mathcal{B}}$ increases, albeit at the cost of lower throughput in the benign regime, reflecting the conservative nature of its trimming strategy. At larger Byzantine fractions, this conservatism enables Su--Vaidya to become more competitive than ADMM and MARL/GNN, although it remains below the proposed ABM framework prior to the breakdown point. Panel~(c) sweeps macro connectivity: the proposed ABM framework dominates throughout with a 1.0--1.3~bits/s/Hz throughput advantage, consistent with Corollary~6, which links higher $\lambda_2$ to improved resilience.
\vspace{-4mm}
\section{Conclusion}
\label{sec:conc}
We have proposed an agent-based-modeling decomposition of distributed resource allocation in 6G networks that separates operator-controllable macro choices and micro knobs from emergent macro outcomes, bridged by three structural metric classes. The macro--micro design rule (Thm.~\ref{thm:rule}) factors $f_{\max}$ into a macro term in $(\lambda_2(\mathbf{L}_f),\overline{\mathrm{FSS}^{\mathrm{trust}}})$ and a micro term in $(\beta,\eta_\tau)$; the fungibility--resilience monotonicity (Prop.~\ref{prop:mono}) couples the design rule to the underlying fungibility landscape; the connectivity--substitutability duality (Thm.~\ref{thm:duality}) gives a closed-form macro-resource exchange rate; the multi-application generalization (Prop.~\ref{prop:apps}) lifts the framework to ISAC, federated learning, and V2X. A natural next step is extending the framework to non-stationary fungibility landscapes that evolve with deployment conditions, such as changing traffic demands, network topology variations, resource availability fluctuations, and time-varying threat environments.
\bibliographystyle{IEEEtran}
\bibliography{NL}
\end{document}